\documentclass[envcountsect,envcountsame,runningheads,a4paper]{llncs}

\usepackage{grd}
\usepackage{cite}
\usepackage{url}
\usepackage{lmodern}
\usepackage{xcolor}

\hypersetup{
    colorlinks,
    linkcolor={red!50!black},
    citecolor={blue!50!black},
    urlcolor={blue!80!black}
}

\renewcommand{\Intchoice}{\bigsqcap}

\newcommand{\thmeqref}[2]{\ref{#1}-\ref{#2}}

\makeatletter
\@addtoreset{equation}{theorem}
\@addtoreset{equation}{definition}
\makeatother

\allowdisplaybreaks

\begin{document}
\title{Calculational Verification of Reactive Programs \\ with Reactive Relations and Kleene Algebra}
\titlerunning{Calculational Verification of Reactive Programs}

\author{Simon Foster$^{\text{\href{https://orcid.org/0000-0002-9889-9514}{ORCiD}}}$ \and Kangfeng Ye \and Ana Cavalcanti \and Jim Woodcock}
\institute{University of York \\ \email{simon.foster@york.ac.uk}}

\authorrunning{Simon Foster \and Kangfeng Ye \and Ana Cavalcanti \and Jim Woodcock}

\maketitle

\begin{abstract}
  Reactive programs are ubiquitous in modern applications, and so verification is highly desirable. We present a
  verification strategy for reactive programs with a large or infinite state space utilising algebraic laws for reactive
  relations. We define novel operators to characterise interactions and state updates, and an associated equational
  theory. With this we can calculate a reactive program's denotational semantics, and thereby facilitate automated
  proof. Of note is our reasoning support for iterative programs with reactive invariants, which is supported by Kleene
  algebra. We illustrate our strategy by verifying a reactive buffer. Our laws and strategy are mechanised in
  Isabelle/UTP, which provides soundness guarantees, and practical verification support.
\end{abstract}


\section{Introduction}
\label{sec:intro}
Reactive programming~\cite{Harel1985,Bainomugisha2013} is a paradigm that enables effective description of software
systems that exhibit both internal sequential behaviour and event-driven interaction with a concurrent party. Reactive
programs are ubiquitous in safety-critical systems, and typically have a very large or infinite state space. Though
model checking is an invaluable verification technique, it exhibits inherent limitations with state explosion and
infinite-state systems that can be overcome by supplementing it with theorem proving.

Previously~\cite{Foster17c}, we have shown how \emph{reactive contracts} support automated proof. They follow the
design-by-contract paradigm~\cite{Meyer92}, where programs are accompanied by pre- and postconditions. Reactive programs
are often non-terminating and so we also capture intermediate behaviours, where the program has not terminated, but is
quiescent and offers opportunities to interact. Our contracts are triples, $\rc{\!P_1\!}{\!P_2\!}{\!P_3\!}$, where $P_1$
is the precondition, $P_3$ the postcondition, and $P_2$ the ``pericondition''. $P_2$ characterises the quiescent
observations in terms of the interaction history, and the events enabled at that point.

Reactive contracts describe communication and state updates, so $P_1$, $P_2$, and $P_3$ can refer to both a trace
history of events and internal program variables. They are, therefore, called ``reactive relations'': like relations
that model sequential programs, they can refer to variables before ($x$) and later ($x'$) in execution, but also the
interaction trace ($\trace$), in both intermediate and final observations.


Verification using contracts employs refinement ($\refinedby$), which requires that an implementation weakens the
precondition, and strengthens both the peri- and postcondition when the precondition holds. We employ the
``programs-as-predicates'' approach~\cite{Hehner93}, where the implementation ($Q$) is itself denoted as a composition
of contracts. Thus, a verification problem, $\rc{\!P_1\!}{\!P_2\!}{\!P_3\!} \refinedby Q$, can be solved by calculating
a program $\rc{\!Q_1\!}{\!Q_2\!}{\!Q_3\!} = Q$, and then discharging three proof obligations: (1) $Q_1 \refinedby P_1$;
(2) $P_2 \refinedby (Q_2 \land P_1)$; and (3) $P_3 \refinedby (Q_3 \land P_1)$. These can be further decomposed, using
relational calculus, to produce verification conditions. In \cite{Foster17c} we employ this strategy in an Isabelle/HOL
tactic.

For reactive programs of a significant size, these relations are complex, and so the resulting proof obligations are
difficult to discharge using relational calculus. We need, first, abstract patterns so that the relations can be
simplified. This necessitates bespoke operators that allow us to concisely formulate the different kinds of
observation. Second, we need calculational laws to handle iterative programs, which are only partly handled in our
previous work~\cite{Foster17c}.



In this paper we present a novel calculus for description, composition, and simplification of reactive relations in the
stateful failures-divergences model~\cite{Hoare&98,Oliveira&09}. We characterise conditions, external interactions, and
state updates. An equational theory allows us to reduce pre-, peri-, and postconditions to compositions of these atoms
using operators of Kleene algebra~\cite{Kozen90} (KA) and utilise KA proof techniques. Our theory is characterised in
the Unifying Theories of Programming~\cite{Hoare&98,Cavalcanti&06} (UTP) framework. For that, we identify a class of UTP
theories that induce KAs, and utilise it in derivation of calculational laws for iteration. We use our UTP
mechanisation, called Isabelle/UTP~\cite{Foster16a}, to implement an automated verification approach for infinite-state
systems with rich data structures.


The paper is structured as follows. \S\ref{sec:prelim} outlines preliminary material. \S\ref{sec:utp-kleene} identifies
a class of UTP theories that induce KAs, and applies this for calculation of iterative contracts. \S\ref{sec:circus-rc}
specialises reactive relations with new atomic operators to capture stateful failures-divergences, and derives their
equational theory. \S\ref{sec:ext-choice} extends this with support for calculating external choices. \S\ref{sec:iter}
completes the theoretical picture with while loops and reactive invariants. \S\ref{sec:verify} demonstrates the
resulting proof strategy in a small verification. \S\ref{sec:concl} outlines related work and concludes. All our
theorems have been mechanically verified in Isabelle/UTP\footnote{All proofs can be found in the cited series of
Isabelle/HOL reports. For historical reasons, we use the syntax $\ckey{R}_s(P \vdash Q \diamond R)$ in our mechanisation
for a contract $\rc{\!P\!}{\!Q\!}{\!R\!}$, which builds on Hoare and He's original syntax for the theory of
designs~\cite{Hoare&98}.}~\cite{Foster16a,Foster-KA-UTP,Foster-RDES-UTP,Foster-SFRD-UTP}.


\section{Preliminaries}
\label{sec:prelim}
\textbf{Kleene Algebras}~\cite{Kozen90} (KA) characterise sequential and iterative behaviour in nondeterministic
programs using a signature $(K, +, 0, \cdot, 1, {}\star)$, where $+$ is a choice operator with unit $0$, and $\cdot$ a
composition operator, with unit $1$. Kleene closure $P{\star}$ denotes iteration of $P$ using $\cdot$ zero or more
times. We consider the class of weak Kleene algebras~\cite{Guttman2010}, which build on weak dioids.


\begin{definition}
  A weak dioid is a structure $(K, +, 0, \cdot, 1)$ such that
  $(S, +, 0)$ is an idempotent and commutative monoid; $(S, \cdot, 1)$
  is a monoid; $\cdot$ left- and right-distributes over $+$; and $0$
  is a left annihilator for $\cdot$.
\end{definition}
The $0$ operator represents miraculous behaviour. It is a left annihilator of composition, but not a right annihilator
as this often does not hold for programs. $K$ is partially ordered by $x \le y \defs (x + y = y)$, which is defined in
terms of $+$, and has least element $0$. A weak KA extends this with the behaviour of the star.
\begin{definition} \label{def:wka}
  A weak Kleene algebra is a structure $(K\!, +, 0, \cdot, 1, \star)$ such that

  \vspace{-3.4ex}
  \begin{center}
  \begin{tabular}{ll}
  \begin{minipage}{0.5\textwidth}
  \begin{enumerate}
    \item $(K, +, 0, \cdot, 1)$ is a weak dioid
    \item $1+x \cdot x{\star} \le x{\star}$
  \end{enumerate}
  \end{minipage} &
  \begin{minipage}{0.5\textwidth}
  \begin{enumerate}
    \setcounter{enumi}{2}
    \item $z + x \cdot y \le y \implies x{\star} \cdot z \le y$ \label{thm:starinductr}
    \item $z + y \cdot x \le y \implies z \cdot x{\star} \le y$
  \end{enumerate}
  \end{minipage}  
  \end{tabular}
  \end{center}
\end{definition}
Various enrichments and specialisations of these axioms exist; for a complete survey see~\cite{Kozen90}. For our
purposes, these axioms alone suffice. From this base, a number of useful identities can be derived, some of which are
listed below.

\begin{theorem} \label{thm:kalaws} $x{{\star}{\star}} = x{\star}$ ~~ $x{\star} = 1 + x\cdot x\star$ ~~ $(x+y){\star} = (x{\star} \cdot y{\star}){\star}$ ~~ $x \cdot x{\star} = x{\star} \cdot x$ \end{theorem}
%
%
\noindent\textbf{UTP}~\cite{Hoare&98,Cavalcanti&06} uses the ``programs-as-predicate'' approach to encode denotational
semantics and facilitate reasoning about programs. It uses the alphabetised relational calculus, which combines
predicate calculus operators like disjunction ($\lor$), complement ($\neg$), and quantification ($\exists x \!@ P(x)$),
with relation algebra, to denote programs as binary relations between initial variables ($x$) and their subsequent
values ($x'$). The set of relations $\Rel$ is partially ordered by refinement $\refinedby$ (refined-by), denoting
universally closed reverse implication, where $\false$ refines every relation. Relational composition ($\relsemi$)
denotes sequential composition with identity $\II$. We summarise the algebraic properties of relations below.
\begin{theorem} $(Rel, \mathrel{\sqsupseteq}, \false, \relsemi, \II)$ is a Boolean quantale~\cite{Moller2006}, so that:
\begin{enumerate}
  \item $(Rel, \refinedby)$ is a complete lattice, with infimum $\bigvee$, supremum $\bigwedge$, greatest element
     $\false$, least element $\true$, and weakest (least) fixed-point operator $\mu F$;
  \item $(Rel, \lor, \false, \land, \true, \neg)$ is a Boolean algebra;
  \item $(Rel, \relsemi, \II)$ is a monoid with $\false$ as left and right annihilator;
  \item $\relsemi$ distributes over $\bigvee$ from the left and right.
\end{enumerate}
\end{theorem}
We often use $\Intchoice_{i \in I}\, P(i)$ to denote an indexed disjunction over $I$, which intuitively refers to a
nondeterministic choice. Note that the partial order $\le$ of the Boolean quantale is $\mathrel{\sqsupseteq}$, and so
our lattice operators are inverted: for example, $\bigvee$ is the infimum with respect to $\refinedby$, and $\mu F$ is
the least fixed-point.

Relations can directly express sequential programs, whilst enrichment to characterise more advanced paradigms --- such
as object orientation~\cite{SCS06}, real-time~\cite{Sherif2010}, and concurrency~\cite{Hoare&98} --- can be achieved
using UTP theories. A UTP theory is characterised as the set of fixed-points of a function
$\healthy{H} : \Rel \to \Rel$, called a healthiness condition. If $P$ is a fixed-point of $\healthy{H}$ it is said to be
$\healthy{H}$-healthy, and the set of healthy relations is $\theoryset{H} \defs \{ P | \healthy{H}(P) = P \}$. In UTP,
it is desirable that $\healthy{H}$ is idempotent and monotonic so that $\theoryset{H}$ forms a complete lattice under
$\refinedby$, and thus reasoning about both nondeterminism and recursion is possible.

Theory engineering and verification of programs using UTP is supported by Isabelle/UTP~\cite{Foster16a}, which provides
a shallow embedding of the relational calculus on top of Isabelle/HOL, and various approaches to automated proof. In
this paper, we use a UTP theory to characterise reactive programs.

\vspace{1ex}

\noindent\textbf{Reactive Programs.} Whilst sequential programs determine the relationship between an initial
and final state, reactive programs also pause during execution to interact with the environment. For example, the
CSP~\cite{Hoare85,Cavalcanti&06} and \Circus~\cite{Oliveira&09} languages can model networks of concurrent processes
that communicate using shared channels. Reactive behaviour is described using primitives such as event prefix
$a\!\then\!P$, which awaits event $a$ and then enables $P$; conditional guard, $b \guard P$, which enables $P$ when $b$
is true; external choice $P\!\extchoice\!Q$, where the environment resolves the choice by communicating an initial event
of $P$ or $Q$; and iteration $\ckey{while}~b~\ckey{do}~P$.  Channels can carry data, and so events can take the form of
an input ($c?x$) or output ($c!v$). \Circus processes also have local state variables that can be assigned ($x :=
v$). We exemplify \Circus with an unbounded buffer.
\begin{example} \label{ex:buffer} In the $Buffer$ process below, variable $bf : \seq \nat$ records the elements, and
  channels $inp(n : \nat)$ and $outp(n : \nat)$ represent inputs and outputs.
  $$Buffer ~\defs~ bf := \langle\rangle \relsemi \left(
  \begin{array}{l}
    \ckey{while}~~true~~\ckey{do} \\
    ~~\left(
    \begin{array}{l}
      inp?v \then bf := bf \cat \langle v \rangle \\
      \extchoice (\#bf > 0) \guard out!(head(bf)) \then bf := tail(bf)
    \end{array} 
    \right)
  \end{array} \right)$$
Variable $bf$ is set to the empty sequence $\langle\rangle$, and then a non-terminating loop describes 
the main behaviour. Its body repeatedly allows the environment to either provide a value $v$ over $inp$, followed by 
which $bf$ is extended, or else, if the buffer is non-empty, receive the value at the head, and then $bf$ is contracted. \qed
\end{example}
The semantics of such programs can be captured using reactive contracts~\cite{Foster17c}:
$$\rc{P_1(\trace, \state, r)}{P_2(\trace, \state, r, r')}{P_3(\trace, \state, \state', r, r')}$$
Here, $P_{1\cdots3}$ are reactive relations that respectively encode, (1) the precondition in terms of the initial state
and permissible traces; (2) permissible intermediate interactions with respect to an initial state; and (3) final states
following execution. Pericondition $P_2$ and postcondition $P_3$ are both within the ``guarantee'' part of the
underlying design contract, and so must be strengthened by refinment; see Appendix~\ref{sec:appendix} and
\cite{Foster17c} for details. $P_2$ does not refer to intermediate state variables since they are concealed when a
program is quiescent.

Variable $\trace$ refers to the trace, and $\state, \state' : \Sigma$ to the state, for state space $\Sigma$. Traces are
equipped with operators for the empty trace $\snil$, concatenation $tt_1 \cat tt_2$, prefix $tt_1 \le tt_2$, and
difference $tt_1 - tt_2$, which removes a prefix $tt_2$ from $tt_1$. Technically, $\trace$ is not a relational variable,
but an expression $\trace \defs tr' - tr$ where $tr, tr'$, as usual in UTP, encode the trace
relationally~\cite{Hoare&98}. Nevertheless, due to our previous results~\cite{Foster16a,Foster17b}, $\trace$ can be
treated as a variable. Here, traces are modelled as finite sequences, $\trace : \seq \textit{Event}$, for some event
set, though other models are also admitted~\cite{Foster17b}. Events can be parametric, written $a.x$, where $a$ is a
channel and $x$ is the data. Moreover, the relations can encode additional semantic data, such as refusals, using
variables $r, r'$. Our theory, therefore, provides an extensible denotational semantic model for reactive and concurrent
languages.

To exemplify, we consider the event prefix and assignment operators from \Circus, which require that we add variable
$ref' : \power(\textit{Event})$ to record refusals.
%
\begin{align*}
  a \then \Skip ~~\defs~~~& \rc{\truer}{\trace = \langle\rangle \land a \notin ref'}{\trace = \langle a \rangle \land \state' = \state} \\
  x := v ~~\defs~~~& \rc{\truer}{\false}{\state' = \state \oplus \{x \mapsto v\} \land \trace =\langle\rangle}
\end{align*}
%
\noindent Prefix has a true precondition, indicated using the reactive relation $\truer$, since the environment cannot
cause errors. In the pericondition, no events have occurred ($\trace = \langle\rangle$), but $a$ is not being
refused. In the postcondition, the trace is extended by $a$, and the state is unchanged. Assignment also has a true
precondition, but a false pericondition since it terminates without interaction. The postcondition updates the state,
and leaves the trace unchanged.


Reactive relations and contracts are characterised by healthiness conditions $\healthy{RR}$ and $\healthy{NSRD}$,
respectively, which we have previously described~\cite{Foster17c}, and reproduce in
Appendix~\ref{sec:appendix}. $\healthy{NSRD}$ specialises the theory of reactive
designs~\cite{Cavalcanti&06,Oliveira&09} to \emph{normal stateful reactive designs}~\cite{Foster17c}. Both
$\theoryset{RR}$ and $\theoryset{NSRD}$ are closed under sequential composition, and have units $\IIr$ and $\IIsrd$,
respectively. Both also form complete lattices under $\refinedby$, with top elements $\false$ and
$\Miracle = \rc{\!\truer\!}{\!\false\!}{\!\false\!}$, respectively. $\Chaos = \rc{\!\false\!}{\!\false\!}{\!\false\!}$,
the least determinisitic contract, is the bottom of the reactive contract lattice. We define the conditional operator
$\conditional{P\!}{\!b\!}{\!Q} \defs ((b\!\land\!P)\!\lor\!(\neg b\!\land\!Q))$, where $b$ is a condition on unprimed
state variables, which can be used for both reactive relations and contracts. We then define the state test operator
$\rasm{b} \defs \conditional{\IIr\!}{\!b\!}{\!\false}$.


Contracts can be composed using relational calculus. The following identities~\cite{Foster17c,Foster-RDES-UTP} show how
this entails composition of the underlying pre-, peri-, and postconditions for $\bigsqcap$ and $\relsemi$, and also
demonstrates closure under these operators.

\begin{theorem}[Reactive Contract Composition] \label{thm:rc-comp}
\begin{align}
  \textstyle \bigsqcap_{i\in I} \, \rc{\!P(i)\!}{\!Q(i)\!}{\!R(i)\!} &= \textstyle \rc{\bigwedge_{i \in I} P(i)}{\bigvee_{i \in I} Q(i)}{\bigvee_{i \in I} R(i)} \label{thm:rc-choice} \\[.1ex]
  \rc{\!P_1\!}{\!P_2\!}{\!\!P_3\!} \relsemi \rc{\!Q_1\!}{\!Q_2\!}{\!\!Q_3\!}&= \rc{\!P_1\!\land\!(P_3\!\wpR\!Q_1)\!}{P_2\!\lor\!(P_3\!\relsemi\!Q_2)\!}{\!P_3\!\relsemi\!Q_3} \label{thm:rc-seq} 
\end{align}
\end{theorem}
%
Nondeterministic choice requires all preconditions, and asserts that one of the peri- and postcondition pairs hold. For
sequential composition, the precondition assumes that $P_1$ holds, and that $P_3$ fulfils $Q_1$. The latter is
formulated using a reactive weakest precondition ($\!\wpR\!$), which obeys standard laws~\cite{Dijkstra75} such as:
$$\textstyle(\bigvee_{i \in I}\,P(i))\!\wpR\!R = \bigwedge_{i \in I} \, P(i)\!\wpR\!R \qquad (P\!\relsemi\!Q)\!\wpR\!R = P\!\wpR\!(Q\!\wpR\!R)$$
In the pericondition, either the first contract is intermediate ($P_2$), or else it terminated ($P_3$) and then
following this the second is intermediate ($Q_2$). In the postcondition the contracts have both terminated in sequence
($P_3 \relsemi Q_3$).


With these and related theorems~\cite{Foster16a}, we can calculate contracts of reactive programs. Verification, then,
can be performed by proving a refinement between two reactive contracts, a strategy we have mechanised in the
Isabelle/UTP tactics \textsf{rdes-refine} and \textsf{rdes-eq}~\cite{Foster16a}. The question remains, though, how to
reason about the underlying compositions of reactive relations for the \text{pre-}, peri-, and postconditions. For
example, consider the action $(a\!\then\!\Skip) \relsemi x\!:=\!v$. For its postcondition, we must simplify
$(\trace = \langle a \rangle \land \state' = \state) \relsemi (\state' = \state \oplus \{x \mapsto v\} \land \trace =
\langle\rangle)$.
In order to simplify its precondition, we also need to consider reactive weakest preconditions. Without such
simplifications, reactive relations can grow very quickly and hamper proof. Finally, of particular importance is the
handling of iterative reactive relations. We address these issues in this paper.


\section{Linking UTP and Kleene Algebra}
\label{sec:utp-kleene}
In this section, we characterise properties of a UTP theory sufficient to identify a KA, and use this to obtain theorems for
iterative contracts. We observe that UTP relations form a KA $(Rel, \intchoice, \relsemi, \star, \II)$, where
$P\star \defs (\nu X @ \II \intchoice P \relsemi X)$. We have proved this definition equivalent to the power form:
$P\star = (\Intchoice_{i \in \nat} ~ P^i)$ where $P^n$ iterates sequential composition $n$ times.

Typically, UTP theories, like $\theoryset{NSRD}$, share the operators for choice ($\intchoice$) and composition
($\relsemi$), only redefining them when absolutely necessary. Formally, given a UTP theory defined by a healthiness
condition $\healthy{H}$, the set of healthy relations $\theoryset{H}$ is closed under $\intchoice$ and $\relsemi$. This
has the major advantage that a large body of laws is directly applicable from the relational calculus. The ubiquity of
$\intchoice$, in particular, can be characterised through the subset of continuous UTP theories, where $\healthy{H}$
distributes through arbitrary non-empty infima, that is,
\begin{center}
$\healthy{H}\left(\Intchoice_{i\in I}\,P(i)\right) = \Intchoice_{i\in I}\,\healthy{H}(P(i)) \text{ provided } I \neq \emptyset$.
\end{center}
Monotonicity of $\healthy{H}$ follows from continuity, and so such theories induce a complete lattice. Continuous UTP
theories include designs~\cite{Hoare&98,Guttman2010}, CSP, and \Circus~\cite{Oliveira&09}. A further consequence of
continuity is that the relational weakest fixed-point operator $\mu X @ F(X)$ constructs healthy relations when
$F : Rel \to \theoryset{H}$.

Though these theories share infima and weakest fixed-points, they do not, in general, share $\top$ and $\bot$ elements,
which is why the infima are non-empty in the above continuity property. Rather, we have a top element
$\thtop{H} \defs \healthy{H}(\false)$ and a bottom element $\thbot{H} \defs \healthy{H}(\true)$~\cite{Foster17c}. The
theories also do not share the relational identity $\II$, but typically define a bespoke identity $\IIT{H}$, which means
that $\theoryset{H}$ is not closed under the relational Kleene star. However, $\theoryset{H}$ is closed under the
related Kleene plus $P^{+} \defs P \relsemi P\star$ since it is equivalent to $(\Intchoice_{i \in \nat} ~ P^{i+1})$,
which iterates $P$ one or more times. Thus, we can obtain a theory Kleene star with the definition
$P \bm{\star} \defs \IIT{H} \sqcap P^{+}$, under which $\healthy{H}$ is indeed closed. We, therefore, define the
following criteria for a UTP theory.
\begin{definition}
  A Kleene UTP theory $(\healthy{H}, \IIT{H})$ satisfies the following conditions: (1) $\healthy{H}$ is idempotent and
  continuous; (2) $\healthy{H}$ is closed under sequential composition; (3) identity $\IIT{H}$ is $\healthy{H}$-healthy; (4)
  $\IIT{H} \relsemi P = P \relsemi \IIT{H} = P$, when $P$ is $\healthy{H}$-healthy; (5)
  $\thtop{H} \relsemi P = \thtop{H}$, when $P$ is $\healthy{H}$-healthy.
\end{definition}
From these properties, we can prove the following theorem.
\begin{theorem} \label{thm:kautp} If $(\healthy{H}, \IIT{H})$ is a Kleene UTP theory, then
  $(\theoryset{H}, \intchoice, \thtop{H}, \relsemi, \IIT{H}, \bm{\star})$ forms a weak Kleene algebra.
\end{theorem}
\begin{proof}
  We prove this in Isabelle/UTP by lifting of laws from the Isabelle/HOL KA
  hierarchy~\cite{Armstrong2015,Gomes2016}. For details see~\cite{Foster-KA-UTP}.
\end{proof}
All the identities of Theorem~\ref{thm:kalaws} hold in a Kleene UTP theory, thus providing reasoning capabilities for
iterative programs. In particular, we can show that
$(\theoryset{NSRD}, \intchoice, \Miracle, \relsemi, \IIsrd, \bm{\star})$ and
$(\theoryset{RR}, \intchoice, \false, \relsemi, \IIr, \bm{\star})$ both form weak KAs. Moreover, we can now
also show how to calculate an iterative contract~\cite{Foster-RDES-UTP}.

\begin{theorem}[Reactive Contract Iteration] \label{thm:rc-iter}
\begin{align*}
    \rc{P}{Q}{R}\bm{\star} &= \rc{R\bm{\star} \wpR P}{R\bm{\star} \relsemi Q}{R\bm{\star}}
\end{align*}
\end{theorem}
Note that the outer and inner star are different operators. The precondition states that $R$ must not violate $P$ after
any number of iterations. The pericondition has $R$ iterated followed by $Q$ holding, since the final observation is
intermediate. The postcondition simply iterates $R$. Thus we have the basis for calculating and reasoning about
iterative contracts.

\section{Reactive Relations of Stateful Failures-Divergences}
\label{sec:circus-rc}
Here, we specialise our contract theory to incorporate failure traces, which are used in CSP, \Circus, and related
languages~\cite{Zhan2008}. We define atomic operators to describe the underlying reactive relations, and the associated
equational theory to expand and simplify compositions arising from Theorems~\ref{thm:rc-comp} and \ref{thm:rc-iter}, and
thus support automated reasoning. We consider external choice separately (\S\ref{sec:ext-choice}).


Healthiness condition $\healthy{NCSP} \defs \healthy{NSRD} \circ \healthy{CSP3} \circ \healthy{CSP4}$ characterises the
stateful failures-divergences model~\cite{Hoare&98,Cavalcanti&06,Oliveira&09}. $\healthy{CSP3}$ and $\healthy{CSP4}$
ensure the refusal sets are well-formed~\cite{Cavalcanti&06,Hoare&98}: $ref'$ can only be mentioned in the pericondition
(see also Appendix~\ref{sec:appendix}). \healthy{NCSP}, like \healthy{NSRD}, is continuous and has $\Skip$, defined
below, as a left and right unit. Thus, $(\theoryset{NCSP}, \intchoice, \Miracle, \relsemi, \Skip, \bm{\star})$ forms a
Kleene algebra. An \healthy{NCSP} contract has the following specialised form~\cite{Foster-SFRD-UTP}.
$$\rc{P(\trace, \state)}{Q(\trace, \state, ref')}{R(\trace, \state, \state')}$$
The underlying reactive relations capture a portion of the stateful failures-divergences. $P$ captures the initial
states and traces that do not induce divergence, that is, unpredictable behaviour like $\Chaos$. $Q$ captures the
stateful failures of a program: the set of events not being refused ($ref'$) having performed trace $\trace$, starting
in state $\state$. $R$ captures the terminated behaviours, where a final state is observed but no refusals. We describe
the pattern of the underlying reactive relations using the following constructs.

\begin{definition}[Reactive Relational Operators]
\begin{align}
  \cspin{b(\state)}{t(\state)} \defs~~ & \healthy{RR}(b(\state) \land t(\state) \le \trace) \\[.2ex]
  \cspen{b(\state)}{t(\state)}{E(\state)} \defs~~ & \healthy{RR}(b(\state) \land \trace = t(\state) \land (\forall e\!\in\!E(\state) @ e \notin ref')) \\[.2ex]
  \csppf{b(\state)}{\sigma}{t(\state)} \defs~~ & \healthy{RR}(b(\state) \land \state' = \sigma(\state) \land \trace = t(\state))
\end{align}
\end{definition}
\noindent In this definition, we utilise expressions $b$, $t$, and $E$ that refer only to the variables by which they
are parametrised. Namely, $b(\state)\!:\!\Bool$ is a condition on $\state$, $t(\state)\!:\!\seq \textit{Event}$ is a
trace expression that describes a possible event sequence in terms of $\state$, and $E(\state) : \power \textit{Event}$
is an expression that describes a set of events. Following \cite{Back1998}, we describe state updates with substitutions
$\sigma : \Sigma \to \Sigma$. We use $\llparenthesis x \mapsto v \rrparenthesis$ to denote a substitution, which is the
identity for every variable, except that $v$ is assigned to $x$. Substitutions can also be applied to contracts and
relations using operator $\substapp{\sigma}{P}$, and then
$Q[v/x] \defs \substapp{\llparenthesis x \mapsto v \rrparenthesis}{Q}$. This operator obeys similar laws to syntactic
substitution, though it is a semantic operator~\cite{Foster16a}.

$\cspin{b(\state)}{t(\state)}$ is a specification of initial behaviour used in preconditions. It states that initially
the state satisfies condition $b$, and $t$ is a prefix of the overall trace. $\cspen{b(\state)}{t(\state)}{E(\state)}$
is used in periconditions to specify quiescent observations, and corresponds to a failure trace. It specifies that the
state variables initially satisfy $b$, the interaction described by $t$ has occurred, and finally we reach a quiescent
phase where none of the events in $E$ are being refused. $\csppf{b(\state)}{\sigma}{t(\state)}$ is used to encode final
terminated observations in the postcondition. It specifies that the initial state satisfies $b$, the state update
$\sigma$ is applied, and the interaction $t$ has occurred.

These operators are all deterministic, in the sense that they describe a single interaction and state-update
history. There is no need for explicit nondeterminism here, as this is achieved using $\bigvee$. These operators allow
us to concisely specify the basic operators of our theory as given below.

\begin{definition}[Basic Reactive Operators] \label{thm:bcircus-def}
\begin{align}
  \assignsC{\sigma} ~\defs~& \rc{\truer}{\false}{\csppf{\ptrue}{\sigma}{\langle\rangle}} \label{def:gen-asn} \\[.1ex]
  \ckey{Do}(a)      ~\defs~& \rc{\truer}{\cspen{\ptrue}{\langle\rangle}{\{a\}}}{\csppf{\ptrue}{id}{\langle a \rangle}} \\[.1ex]
  \Stop             ~\defs~& \rc{\truer}{\cspen{\ptrue}{\langle\rangle}{\emptyset}}{\false} 
\end{align}
\end{definition}


\noindent Generalised assignment $\assignsC{\sigma}$ is again inspired by \cite{Back1998}. It has a $\truer$
precondition and a $\false$ pericondition: it has no intermediate observations. The postcondition states that for any
initial state ($true$), the state is updated using $\sigma$, and no events are produced ($\langle\rangle$). A singleton
assignment $x := v$ can be expressed using a state update $\llparenthesis x \mapsto v \rrparenthesis$. We define
$\Skip \defs \assignsC{id}$, which leaves all variables unchanged.

$\ckey{Do}(a)$ encodes an event action. Its pericondition states that no event has occurred, and $a$ is accepted. Its
postcondition extends the trace by $a$, leaving the state unchanged. We can denote \Circus event prefix $a \then P$ as
$\ckey{Do}(a) \relsemi P$. 

Finally, $\Stop$ represents a deadlock: its pericondition states the trace is unchanged and no events are being
accepted. The postcondition is false as there is no way to terminate. A \Circus guard $g \guard P$ can be denoted as
$(\conditional{P}{g}{\Stop})$, which behaves as $P$ when $g$ is true, and otherwise deadlocks.

To calculate contractual semantics, we need laws to reduce pre-, peri-, and postconditions. These need to cater for
various composition cases of operators $\bigsqcap$, $\relsemi$, and $\Extchoice$. So, we prove~\cite{Foster-SFRD-UTP}
the following composition laws for $\mathcal{E}$ and $\Phi$.

\begin{theorem}[Reactive Relational Compositions] \label{thm:crel-comp}
  \begin{align}
    \rasm{b} \relsemi P =~~& \csppf{b}{id}{\langle\rangle} \relsemi P \label{thm:crc1} \\[.2ex]
    \csppf{b_1}{\sigma_1}{t_1} \relsemi \csppf{b_2}{\sigma_2}{t_2} =~~& \csppf{b_1 \land \substapp{\sigma_1}{b_2}}{\sigma_2 \circ \sigma_1}{t_1 \cat \substapp{\sigma_1}{t_2}} \label{thm:crc2} \\[.2ex]
    \csppf{b_1}{\sigma_1}{t_1} \relsemi \cspen{b_2}{t_2}{E} =~~& \cspen{b_1 \land \substapp{\sigma_1}{b_2}}{t_1 \cat \substapp{\sigma_1}{t_2}}{\substapp{\sigma_1}{E}} \label{thm:crc3} \\[.2ex]
    \conditional{\csppf{b_1}{\sigma_1}{t_1}}{c}{\csppf{b_2}{\sigma_2}{t_2}} =~~& \csppf{\conditional{b_1\!}{\!c\!}{\!b_2}}{\conditional{\sigma_1\!}{\!c\!}{\!\sigma_2}}{\conditional{t_1\!}{\!c\!}{\!t_2}} \label{thm:crc4} \\[.2ex]
    \conditional{\cspen{b_1}{t_1}{E_1}}{c}{\cspen{b_2}{t_2}{E_2}} =~~& \cspen{\conditional{b_1\!}{\!c\!}{\!b_2}}{\conditional{t_1\!}{\!c\!}{\!t_2}}{\conditional{E_1\!}{\!c\!}{\!E_2}} \label{thm:crc5} \\[.2ex]
    \textstyle\left(\bigwedge_{i \in I}\,\cspen{b(i)}{t}{E(i)}\right) =~~& \textstyle\cspen{\bigwedge_{i \in I}\,b(i)}{t}{\bigcup_{i \in I}\,E(i)} \label{thm:crc6}
  \end{align}
\end{theorem}

\noindent Law \eqref{thm:crc1} states that a precomposed test can be expressed using $\Phi$. \eqref{thm:crc2} states
that the composition of two terminated observations results in the conjunction of the state conditions, composition of
the state updates, and concatenation of the traces. It is necessary to apply the initial state update $\sigma_1$ as a
substitution to both the second state condition ($s_2$) and the trace expression ($t_2$). \eqref{thm:crc3} is similar,
but accounts for the enabled events rather than state updates. \eqref{thm:crc2} and \eqref{thm:crc3} are required
because of Theorem~\thmeqref{thm:rc-comp}{thm:rc-seq}, which sequentially composes a pericondition with a postcondition,
and a postcondition with a postcondition. \eqref{thm:crc4} and \eqref{thm:crc5} show how conditional distributes
through the operators. Finally, \eqref{thm:crc6} shows that a conjunction of intermediate observations with a common
trace takes the conjunction of the state conditions, and the union of the enabled events. It is needed for external
choice, which conjoins the periconditions (see \S\ref{sec:ext-choice}).

In order to calculate preconditions, we need to consider the weakest precondition
operator. Theorem~\ref{thm:rc-comp}-\ref{thm:rc-seq} requires that, in a sequential composition $P \relsemi Q$, we need
to show that the postcondition of contract $P$ satisfies the precondition of contract $Q$. Theorem~\ref{thm:crel-comp}
explains how to eliminate most composition operators in a contract's postcondition, but not in general
$\lor$. Postconditions are, therefore, typically expressed as disjunctions of the $\Phi$ operator, and so it suffices to
calculate its weakest precondition using the theorem below.

\begin{theorem} \, $\csppf{s}{\sigma}{t} \wpR P ~=~ (\cspin{s}{t} \implies (\substapp{\sigma}{P})[\trace - t/\trace])$ \label{thm:evwp}
\end{theorem}
In order for $\csppf{s}{\sigma}{t}$ to satisfy reactive condition $P$, whenever we start in the state satisfying $s$ and
the trace $t$ has been performed, $P$ must hold on the remainder of the trace ($\trace - t$), and with the state update
$\sigma$ applied. We can now use these laws, along with Theorem \ref{thm:rc-comp}, to calculate the semantics of
processes, and to prove equality and refinement conjectures, as we illustrate below.

\begin{example} We show that
  $(x\!:=\!1 \relsemi \ckey{Do}(a.x) \relsemi x := x + 2) ~=~ (\ckey{Do}(a.1) \relsemi x\!:=\!3)$. By applying
  Definition \ref{thm:bcircus-def} and Theorems~\ref{thm:rc-comp} (\ref{thm:rc-seq}), \ref{thm:crel-comp},
  \ref{thm:evwp}, both sides reduce to
  $\,\rc{\truer}{\cspen{true}{\langle\rangle}{\{a.1\}}}{\csppf{true}{\{x \mapsto 3\}}{\langle a.1 \rangle}}$, which has
  a single quiescent state, waiting for event $a.1$, and a single final state, where $a.1$ has occurred and state
  variable $x$ has been updated to $3$. We calculate the left-hand side below.
  \begin{align*}
    &(x\!:=\!1 \relsemi \ckey{Do}(a.x) \relsemi x := x + 2) \\[.5ex]
    =& \left(\begin{array}{l}
         \rc{\truer}{\false}{\csppf{\ptrue}{\substmap{x \mapsto 1}}{\langle\rangle}} \relsemi \\[.2ex]
         \rc{\truer}{\cspen{\ptrue}{\langle\rangle}{\{a.x\}}}{\csppf{\ptrue}{id}{\langle a.x \rangle}} \relsemi \\[.2ex]
         \rc{\truer}{\false}{\csppf{\ptrue}{\substmap{x \mapsto x + 1}}{\langle\rangle}} 
       \end{array}\right) & [\text{Def.}~\ref{thm:bcircus-def}] \\[.5ex]
    =& \rc{\truer}{
         \begin{array}{l}
           \csppf{\ptrue}{\substmap{x \mapsto 1}}{\langle\rangle} \relsemi \\ 
           \cspen{\ptrue}{\langle\rangle}{\{a.x\}}
         \end{array}
         }{
         \begin{array}{l}
           \csppf{\ptrue}{\substmap{x \mapsto 1}}{\langle\rangle} \relsemi \\ 
           \csppf{\ptrue}{id}{\langle a.x \rangle} \relsemi \\
           \csppf{\ptrue}{\substmap{x \mapsto x + 2}}{\langle\rangle}
         \end{array}} & \left[
                        \begin{array}{l}
                          \text{Thm.}~\ref{thm:rc-comp}, \\
                          \text{Thm.}~\ref{thm:evwp}
                        \end{array} \right] \\[.5ex]
    =& \rc{\!\truer\!}{\!
         \begin{array}{l}
           \cspen{\ptrue}{\langle\rangle[1/x]}{\{a.x\}[1/x]}
         \end{array}\!
         }{\!
         \begin{array}{l}
           \csppf{\ptrue}{\substmap{x \mapsto 1}}{\langle a.1 \rangle} \relsemi \\
           \csppf{\ptrue}{\substmap{x \mapsto x\!+\!2}}{\langle\rangle}
         \end{array}\!} & \left[
                        \begin{array}{l}
                          \text{Thm.}~\ref{thm:crel-comp}
                        \end{array} \right] \\[.5ex]
    =& \,\rc{\truer}{\cspen{true}{\langle\rangle}{\{a.1\}}}{\csppf{true}{\{x \mapsto 3\}}{\langle a.1 \rangle}} & \square
  \end{align*}
  
\end{example}
\noindent Similarly, we can use our theorems, with the help of our mechanised proof strategy in \textsf{rdes-eq}, to
prove a number of general laws~\cite{Foster-SFRD-UTP}.

\begin{theorem}[Stateful Failures-Divergences Laws] \label{thm:sfdl}

  \vspace{-1ex}
  \noindent
  \begin{minipage}{0.65\textwidth}
  \begin{align}
  \assignsC{\sigma}\relsemi\rc{\!P_1\!}{\!P_2\!}{\!P_3\!} &= \rc{\!\substapp{\sigma\!}{\!P_1}\!}{\!\substapp{\sigma\!}{\!P_2}\!}{\!\substapp{\sigma\!}{\!P_3}\!} \label{thm:asgdist}  \\
  \assignsC{\sigma} \relsemi \ckey{Do}(e) &= \ckey{Do}(\substapp{\sigma}{e}) \relsemi \assignsC{\sigma} \label{thm:asgev}
  \end{align}
  \end{minipage}
  \begin{minipage}{0.35\textwidth}
  \begin{align}
  \assignsC{\sigma} \relsemi \assignsC{\rho} &= \assignsC{\rho \circ \sigma} \label{thm:asgcomp} \\ 
  \Stop \relsemi P &= \Stop \label{thm:dlockzero}
  \end{align}
  \end{minipage}
\end{theorem}

\noindent Law~\eqref{thm:asgdist} shows how assignment distributes substitutions through a contract. \eqref{thm:asgev}
and \eqref{thm:asgcomp} are consequences of~\eqref{thm:asgdist}. \eqref{thm:dlockzero} shows that deadlock is a left
annihilator. 



\section{External Choice and Productivity}
\label{sec:ext-choice}
In this section we consider reasoning about programs with external choice, and characterise the important subclass of
productive contracts~\cite{Foster17c}, which are also essential in verifying recursive and iterative reactive programs.

An external choice $P \extchoice Q$ resolves whenever either $P$ or $Q$ engages in an event or terminates. Thus, its
semantics requires that we filter observations with a non-empty trace. We introduce healthiness condition
$\healthy{R4}(P) \defs (P \land \trace > \langle\rangle)$, whose fixed points strictly increase the trace, and its dual
$\healthy{R5}(P) \defs (P \land \trace = \langle\rangle)$ where the trace is unchanged. We use these to define indexed
external choice.
\begin{definition}[Indexed External Choice] \label{def:ext-choice}
\begin{align*}
& \Extchoice i \in I @ \rc{P_1(i)}{P_2(i)}{P_3(i)} \defs \\[.5ex]
& \qquad \textstyle \rc{\bigwedge_{i \in I} ~ P_1(i)}{\left(\bigwedge_{i \in I} ~ \healthy{R5}(P_2(i))\right) \lor \left(\bigvee_{i \in I} ~ \healthy{R4}(P_2(i))\right)}{\bigvee_{i \in I}{P_3(i)}}
\end{align*}
\end{definition}
This enhances the binary definition~\cite{Hoare&98,Oliveira&09}, and recasts our definition in~\cite{Foster17c} for
calculation. Like nondeterministic choice, the precondition requires that all constituent preconditions are
satisfied. In the pericondition \healthy{R4} and \healthy{R5} filter all observations. We take the conjunction of all
\ckey{R5} behaviours: no event has occurred, and all branches are offering to communicate. We take the disjunction of
all \ckey{R4} behaviours: an event occurred, and the choice is resolved. In the postcondition the choice is resolved,
either by communication or termination, and so we take the disjunction of all constituent postconditions. Since
unbounded choice is supported, we can denote indexed input prefix for any size of input domain $A$:
$$a?x\!:\!A \then P(x) ~~~\defs~~~ \Extchoice x \in A @ a.x \then P(x)$$

\noindent We next show how \healthy{R4} and \healthy{R5} filter the various reactive relational operators, which can be applied to
reason about contracts involving external choice.
\begin{theorem}[Trace Filtering] \label{thm:filtering}

\vspace{1ex}
\begin{tabular}{cc}
$\begin{aligned}
  \textstyle\healthy{R4}\left(\bigvee_{i \in I} P(i)\right) & = \textstyle\bigvee_{i \in I} \healthy{R4}(P(i)) \\[1ex]
  \healthy{R4}(\csppf{s}{\sigma}{\langle \rangle}) & = \false \\[1ex]
  \healthy{R4}(\csppf{s}{\sigma}{\langle a, ... \rangle}) & = \csppf{s}{\sigma}{\langle a, ... \rangle}
\end{aligned}
\quad
\begin{aligned}
  \textstyle\healthy{R5}\left(\bigvee_{i \in I} P(i)\right) & = \textstyle\bigvee_{i \in I} \healthy{R5}(P(i)) \\[1ex]
  \healthy{R5}(\cspen{s}{\langle \rangle}{E}) & = \cspen{s}{\langle \rangle}{E} \\[1ex]
  \healthy{R5}(\cspen{s}{\langle a, ... \rangle}{E}) & = \false
\end{aligned}$
\end{tabular}
\end{theorem}
Both operators distribute through $\bigvee$. Relations that produce an empty trace yield $\false$ under $\healthy{R4}$
and are unchanged under $\healthy{R5}$. Relations that produce a non-empty trace yield $\false$ for $\healthy{R5}$, and
are unchanged underr $\healthy{R4}$. We can now filter the behaviours that do and do not resolve the choice, as
exemplified below.

\begin{example} Consider the calculation of $a\!\then\!b\!\then\!\Skip \extchoice c\!\then\!\Skip$. The left branch has
  two quiescent observations, one waiting for $a$, and one for $b$ having performed $a$: its pericondition is
  $\cspen{true}{\langle \rangle}{\{a\}} \lor \cspen{true}{\langle a \rangle}{\{b\}}$.  Application of \healthy{R5} to
  this will yield the first disjunct, since the trace has not increased, and \healthy{R4} will yield the second
  disjunct. For the right branch there is one quiescent observation, $\cspen{true}{\langle \rangle}{\{c\}}$, which
  contributes an empty trace and is $\healthy{R5}$ only. The overall pericondition is
  $(\cspen{true}{\langle \rangle}{\{a\}} \land \cspen{true}{\langle \rangle}{\{c\}}) \lor \cspen{true}{\langle a
    \rangle}{\{b\}}$,
  which is simply $\cspen{true}{\langle \rangle}{\{a, c\}} \lor \cspen{true}{\langle a \rangle}{\{b\}}$. \qed
\end{example}

\noindent By calculation we can now prove that $(\theoryset{NCSP}, \extchoice, \Stop)$ forms a commutative and
idempotent monoid, and $\Chaos$, the divergent program, is its annihilator. Sequential composition also distributes from
the left and right through external choice, but only when the choice branches are productive~\cite{Foster17c}.

\begin{definition}
  A contract $\rc{P_1}{P_2}{P_3}$ is productive when $P_3$ is \healthy{R4} healthy.
\end{definition}

\noindent A productive contract is one that, whenever it terminates, strictly increases the trace. For example
$a \then \Skip$ is productive, but $\Skip$ is not. Constructs that do not terminate, like $\Chaos$, are also
productive. The imposition of \healthy{R4} ensures that only final observations that increase the trace, or are $\false$,
are admitted.

We define healthiness condition $\healthy{PCSP}$, which extends $\healthy{NCSP}$ with productivity. We also define
$\healthy{ICSP}$, which formalises instantaneous contracts where the postcondition is \healthy{R5} healthy and the
pericondition is $\false$. For example, both $\Skip$ and $x := v$ are $\healthy{ICSP}$ healthy as they do not contribute
to the trace and have no intermediate observations. This allows us to prove the following laws.

\begin{theorem}[External Choice Distributivity]
  \begin{align*}
    (\Extchoice i\!\in\!I @ P(i)) \relsemi Q &~=~ \Extchoice i\!\in\!I @ (P(i) \relsemi Q) & ~~[\text{if, } \forall i\!\in\!I, P(i) \text{ is } \healthy{PCSP} \text{ healthy}] \\[.3ex]
    P \relsemi (\Extchoice i\!\in\!I @ Q(i)) &~=~ \Extchoice i\!\in\!I @ (P \relsemi Q(i)) & ~~[\text{if } P \text{ is } \healthy{ICSP} \text{ healthy}]
  \end{align*}
\end{theorem}

\noindent The first law follows because every $P(i)$, being productive, must resolve the choice before terminating, and
thus it is not possible to reach $Q$ before this occurs. It generalises the standard guarded choice distribution law for
CSP~\cite[page 211]{Hoare&98}. The second law follows for the converse reason: since $P$ cannot resolve the choice with
any of its behaviour, it is safe to execute it first.

Productivity also forms an important criterion for guarded recursion that we utilise in \S\ref{sec:iter} to
calculate fixed points. \healthy{PCSP} is closed under several operators.

\begin{theorem} \textnormal{(1)} $\Miracle$, $\Stop$, $\ckey{Do}(a)$ are \healthy{PCSP}; \textnormal{(2)} $P\!\relsemi\!Q$ is \healthy{PCSP} if either $P$ or $Q$ is \healthy{PCSP}; \textnormal{(3)} $\Extchoice i \in I @ P(i)$ is \healthy{PCSP} if, for all $i\!\in\!I$, $P(i)$ is \healthy{PCSP}.
\end{theorem}
Calculation of external choice is now supported, and a notion of productivity defined. In the next section we use
the latter for calculation of while-loops.

\section{While Loops and Reactive Invariants}
\label{sec:iter}
In this section we complete our verification strategy by adding support for iteration. Iterative programs can be
constructed using the reactive while loop.
$$\while{b}{P} ~~\defs~~ (\mu X @ \conditional{P \relsemi X}{b}{\Skip}).$$ 
\noindent We use the weakest fixed-point so that an infinite loop with no observable activity corresponds to the
divergent action $\Chaos$, rather than $\Miracle$. For example, we can show that
$(\while{true}{~x\!:=\! x\! +\! 1~}) = \Chaos$. The $true$ condition is not a problem in this context because, unlike
its imperative cousin, the reactive while loop pauses for interaction with its environment during execution, and
therefore infinite executions are observable and therefore potentially useful.

In order to reason about such behaviour, we need additional calculational laws. A fixed-point ($\mu X @ F(X)$) is
guarded provided at least one event is contributed to the trace by $F$ prior to it reaching $X$. For instance,
$\mu X @ a \then X$ is guarded, but $\mu X @ y := 1 \relsemi X$ is not. Hoare and He's theorem~\cite[theorem
8.1.13, page 206]{Hoare&98} states that if $F$ is guarded, then there is a unique fixed-point and hence
$(\mu X @ F(X)) = (\nu X @ F(X))$. Then, provided $F$ is continuous, we can invoke Kleene's fixed-point theorem to
calculate $\nu F$. Our previous result~\cite{Foster17c} shows that if $P$ is productive, then $\lambda X @ P \relsemi X$
is guarded, and so we can calculate its fixed-point. We now generalise this for the function above.

\begin{theorem} If $P$ is productive, then $(\mu X @ \conditional{P \relsemi X}{b}{\Skip})$ is guarded.
\end{theorem}
\begin{proof} In addition to our previous theorem~\cite{Foster17c}, we use the following properties:
  \vspace{-1ex}
  \begin{itemize}
    \item If $X$ is not mentioned in $P$ then $\lambda X @ P$ is guarded;
    \item If $F$ and $G$ are both guarded, then $\lambda X @ \conditional{F(X)}{b}{G(X)}$ is guarded. \qed
  \end{itemize} 
\end{proof}
This allows us to convert the fixed-point into an iterative form. In particular, we can prove the following theorem that
expresses it in terms of Kleene star.
\begin{theorem} If $P$ is \healthy{PCSP} healthy then
  $\while{b}{P} = (\rasm{b} \relsemi P)\bm{\star} \relsemi \rasm{\neg b}$.
\end{theorem}
This theorem is similar to the usual imperative definition~\cite{Armstrong2015,Gomes2016}. $P$ is executed multiple
times when $b$ is true initially, but each run concludes when $b$ is false. However, due to the embedding of reactive
behaviour, there is more going on than meets the eye; the next theorem shows how to calculate an iterative contract.

\begin{theorem} If $R$ is \healthy{R4} healthy then \label{thm:whilecalc} 
  \begin{align*}
  \while{b}{\rc{\!P\!}{\!Q\!}{\!R\!}} =& \rc{(\rasm{b}\!\relsemi\!R)\bm{\star} \wpR (b\!\implies\!P)}{(\rasm{b} \relsemi R)\bm{\star}\!\relsemi\!\rasm{b}\!\relsemi\!Q}{(\rasm{b}\!\relsemi\!R)\bm{\star}\!\relsemi\!\rasm{\neg b}}
  \end{align*}
\end{theorem}
The precondition requires that any number of $R$ iterations, where $b$ is initially true, satisfies $P$. This ensures
that the contract does not violate its own precondition from one iteration to the next. The pericondition states that
intermediate observations have $R$ executing several times, with $b$ true, and following this $b$ remains true and
the contract is quiescent ($Q$). The postcondition is similar, but after several iterations, $b$
becomes false and the loop terminates, which is the standard relational form of a while loop.

Theorem~\ref{thm:whilecalc} can be utilised to prove a refinement introduction law for the reactive while loop. This
employs ``reactive invariant'' relations, which describe how both the trace and state variables are permitted to evolve.
\begin{theorem}$\rc{\!I_1\!}{\!I_2\!}{\!I_3\!} \refinedby \while{b}{\rc{\!Q_1\!}{\!Q_2\!}{\!Q_3\!}}$
  provided that: \label{thm:rea-inv}
  \vspace{-.5ex}
  \begin{enumerate} \itemsep4pt
    \item the assumption is weakened $((\rasm{b} \relsemi Q_3)\bm{\star} \wpR (b \implies Q_1) \refinedby I_1)$;
    \item when $b$ holds, $Q_2$ establishes the $I_2$ pericondition invariant $(I_2 \refinedby (\rasm{b} \relsemi Q_2))$ and,
      $Q_3$ maintains it $(I_2 \refinedby \rasm{b} \relsemi Q_3 \relsemi I_2)$;
    \item postcondition invariant $I_3$ is established when $b$ is false ($I_3 \refinedby \rasm{\neg b}$) and $Q_3$
      establishes it when $b$ is true ($I_3 \refinedby \rasm{b} \relsemi Q_3 \relsemi I_3$).
  \end{enumerate}
\end{theorem}
\begin{proof}
  By application of refinement introduction, with Theorems~\thmeqref{def:wka}{thm:starinductr} and \ref{thm:whilecalc}.
\end{proof}
\noindent Theorem~\ref{thm:rea-inv} shows the conditions under which an iterated reactive contract satisfies an
invariant contract $\rc{\!I_1\!}{\!I_2\!}{\!I_3\!}$. Relations $I_2$ and $I_3$ are reactive invariants that must hold in
quiescent and final observations, respectively. Both can refer to $\state$ and $\trace$, $I_2$ can additionally refer to
$ref'$, and $I_3$ to $\state'$. Combined with the results from \S\ref{sec:circus-rc} and \S\ref{sec:ext-choice},
this result provides the basis for a proof strategy for iterative reactive programs that we now exemplify.

\section{Verification Strategy for Reactive Programs}
\label{sec:verify}
Our collected results give rise to an automated verification strategy for iterative reactive programs, whereby we
(1)~calculate the contract of a reactive program, (2)~use our equational theory to simplify the underlying reactive
relations, (3)~identify suitable invariants for reactive while loops, and (4) finally prove refinements using relational
calculus. Although the underlying relations can be quite complex, our equational theory from \S\ref{sec:circus-rc} and
\S\ref{sec:ext-choice}, aided by the Isabelle/HOL simplifier, can be used to rapidly reduce them to more compact forms
amenable to automated proof. In this section we illustrate this strategy with the buffer in Example~\ref{ex:buffer}. We
prove two properties: (1) deadlock freedom, and (2) that the order of values produced is the same as those consumed.

We first calculate the contract of the main loop in the $Buffer$ process and then use this to calculate the overall
contract for the iterative behaviour.

\begin{theorem}[Loop Body]
  The body of the loop is $\rc{\truer}{B_2}{B_3}$ where
  \begin{align*}
    B_2 =~~& \textstyle\cspen{true}{\langle\rangle}{\bigcup_{v \in \nat}\, \{inp.v\} \cup (\conditional{\{out.head(bf)\}}{0 < \#bf}{\emptyset})} \\[1ex]
    B_3 =~~& \left(
      \begin{array}{l}
        \textstyle
        \left(\bigvee_{v \in \nat}\, \csppf{true}{\{bf \mapsto bf \cat \langle v \rangle\}}{\langle inp.v \rangle} \right) \lor \\[1ex]
        \csppf{0 < \#bf}{\{bf \mapsto tail(bf)\}}{\langle out.head(bf) \rangle}
      \end{array}
    \right)
  \end{align*}
  \noindent \textnormal{The $\truer$ precondition implies no divergence. The pericondition states that every input event
    is enabled, and the output event is enabled if the buffer is non-empty. The postcondition contains two
    possible final observations: (1) an input event occurred and the buffer variable was extended; or (2)
    provided the buffer was non-empty initially, then the buffer's head is output and $bf$ is contracted.}
\end{theorem}
\begin{proof} To exemplify, we calculate the left-hand side of the choice, employing Theorems~\ref{thm:rc-comp},
  \ref{thm:bcircus-def}, \ref{thm:crel-comp}, and \ref{thm:filtering}. The entire calculation is automated in Isabelle/UTP.
  \begin{align*}
    & inp?v \then bf := bf \cat \langle v \rangle \\
    & = \Extchoice v\!\in\!\nat @ \ckey{Do}(inp.v) \relsemi bf := bf \cat \langle v \rangle & [\textnormal{Defs}]\\ 
    & = \Extchoice v\!\in\!\nat @ \left(
      \begin{array}{l}
        \rc{\truer}{\cspen{\ptrue}{\langle\rangle}{\{inp.v\}}}{\csppf{\ptrue}{id}{\langle inp.v \rangle}} \relsemi \\ 
        \rc{\truer}{\false}{\csppf{\ptrue}{\llparenthesis bf \mapsto bf \cat \langle v \rangle \rrparenthesis}{\langle\rangle}}
      \end{array} \right) & [\ref{thm:bcircus-def}] \\[.5ex]
    & = \Extchoice v\!\in\!\nat @
        \rc{\!\truer\!}{\!
        \begin{array}{l}
          \cspen{\ptrue}{\langle\rangle}{\{inp.v\}} \\
          \lor \false
        \end{array} \!
      }{\!
        \begin{array}{l}
          \csppf{\ptrue}{id}{\langle inp.v \rangle} \relsemi \\
          \csppf{\ptrue}{\llparenthesis bf\!\mapsto\!bf\!\cat\!\langle v \rangle \rrparenthesis}{\langle\rangle}
        \end{array}
      } & [\ref{thm:rc-comp}, \ref{thm:evwp}] \\[1ex]
    & = \Extchoice\! v\!\in\!\nat @\! \rc{\!\truer\!}{\!\cspen{\ptrue}{\snil}{\{\!inp.v\!\}}\!}{\!\csppf{\ptrue}{\llparenthesis bf\!\mapsto\!bf\!\cat\!\langle v \rangle \rrparenthesis}{\langle inp.v \rangle}\!} & [\ref{thm:crel-comp}] \\[1ex]
    & = \rc{\!\truer\!}{\!\cspen{\ptrue}{\snil}{\!\bigcup_{v \in \nat} \{\!inp.v\!\}}\!\!}{\!\bigvee_{v \in \nat} \!\csppf{\ptrue}{\llparenthesis bf \!\mapsto\! bf\!\cat\!\langle v \rangle \rrparenthesis}{\langle inp.v \rangle}\!} & [\ref{def:ext-choice},\ref{thm:filtering}]
  \end{align*}
\end{proof}

\noindent Though this calculation seems complicated, in practice it is fully automated and thus a user need not be
concerned with these minute calculational details, but can rather focus on finding suitable reactive
invariants. \qed

Then, by Theorem~\ref{thm:whilecalc} we can calculate the overall behaviour of the buffer.
$$Buffer = \rc{\truer}{\csppf{true}{\{bf \mapsto \langle\rangle\}}{\langle\rangle} \relsemi B_3\bm{\star} \relsemi
  B_2}{\false}$$

\noindent This is a non-terminating contract where every quiescent behaviour begins with an empty buffer, performs some
sequence of buffer inputs and outputs accompanied by state updates ($B_3\bm{\star}$), and is finally offering the
relevant input and output events ($B_2$). We can now employ Theorem~\ref{thm:rea-inv} to verify the buffer. First, we
tackle deadlock freedom, which can be proved using the following refinement.

\begin{theorem}[Deadlock Freedom] $$\textstyle\rc{\truer}{\bigvee_{s, t, E, e}~~ \cspen{s}{t}{\{e\} \cup E}}{\truer} \refinedby Buffer$$
\end{theorem}

\noindent Since only quiescent observations can deadlock, we only constrain the pericondition. It
characterises observations where at least one event $e$ is being accepted: there is no deadlock. This theorem can be
discharged automatically in 1.8s on an Intel i7-4790 desktop machine. We next tackle the second property.

\begin{theorem}[Buffer Order Property] The sequence of items output is a prefix of those that were previously
  input. This can be formally expressed as
  $$\rc{\truer}{outps(\trace) \le inps(\trace)}{\truer} \refinedby Buffer$$ where $inps(t), outps(t) : \seq\,\nat$ extract the sequence of
  input and output elements from the trace $t$, respectively. The postcondition is left unconstrained as $Buffer$ does not terminate.
\end{theorem}
  
\begin{proof}
  First, we identify the reactive invariant $I \defs outps(\trace) \le bf \cat inps(\trace)$, and show that
  $\rc{\!\truer\!}{\!I\!}{\!\truer\!} \refinedby \while{\truer}{\rc{\!\truer\!}{\!B_2\!}{\!B_3\!}}$. By
  Theorem~\ref{thm:rea-inv} it suffices to show case (2), that is $I \refinedby B_2$ and $I \refinedby B_3 \relsemi I$,
  as the other two cases are vacuous. These two properties can be discharged by relational calculus. Second, we prove
  that
  $\rc{\!\truer\!}{outps(\trace) \le inps(\trace)}{\!\truer\!} \refinedby bf := \langle\rangle \relsemi
  \rc{\!\truer\!}{\!I\!}{\!\truer\!}$.
  This holds, by Theorem~\thmeqref{thm:sfdl}{thm:asgdist}, since
  $I[\langle\rangle/bf] = outps(\trace) \le inps(\trace)$.  Thus, the overall theorem holds by monotonicity of
  $\relsemi$ and transitivity of $\refinedby$. The proof is semi-automatic --- since we have to manually apply induction
  with Theorem~\ref{thm:rea-inv} --- with the individual proof steps taking 2.2s in total. \qed

\end{proof}

\section{Conclusion}
\label{sec:concl}
We have demonstrated an effective verification strategy for reactive programs employing reactive relations and Kleene
algebra. Our theorems and verification tool can be found in our theory repository\footnote{Isabelle/UTP:
  \url{https://github.com/isabelle-utp/utp-main}}, together with companion proofs.

Related work includes the works of Struth \emph{et al.} on verification of imperative
programs~\cite{Armstrong2015,Gomes2016} using Kleene algebra for verification-condition generation, which our work
heavily draws upon. Automated proof support for the failures-divergences model was previously provided by the CSP-Prover
tool~\cite{Isobe2008}, which can be used to verify infinite-state systems in CSP. Our work is different both in its
contractual semantics, and also in our explicit handling of state, which allows us to express variable assignments. 

Our work also lies within the ``design-by-contract'' field~\cite{Meyer92}. The refinement calculus of reactive
systems~\cite{Preoteasa2014} is a language based on property transformers containing trace information. Like our work,
they support reactive systems that are non-deterministic, non-input-receptive, and infinite state. The main differences
are our handling of state variables, the basis in relational calculus, and our failures-divergences
semantics. Nevertheless, our contract framework~\cite{Foster17c} can be linked to those results, and we expect to derive
an assume-guarantee calculus.

In future work, we will extend our calculation strategy to parallel composition. We aim to apply it to more substantial
examples, and are currently using it to build a prototype tactic for verifying robotic
controllers~\cite{Miyazawa2017}. In this direction, our semantics and techniques will be also be extended to cater for
real-time, probabilistic, and hybrid computational behaviours~\cite{Foster17b}.


\section*{Acknowledgments}

This research is funded by the RoboCalc
project\footnote{RoboCalc Project: \url{https://www.cs.york.ac.uk/circus/RoboCalc/}}, EPSRC
grant EP/M025756/1.

\bibliographystyle{splncs}
\bibliography{RAMICS2018}

\appendix
\section{UTP Theory Definitions}
\label{sec:appendix}
In this appendix, we summarise our theory of reactive design contracts. The definitions are all mechanised in
accompanying Isabelle/HOL reports~\cite{Foster-RDES-UTP,Foster-SFRD-UTP}.

\subsection{Observational Variables}

We declare two sets $\tset$ and $\Sigma$ that denote the sets of traces and state spaces, respectively, and operators
$\tcat : \tset \to \tset \to \tset$ and $\tempty : \tset$. We require that $(\tset, \tcat, \tempty)$ forms a trace
algebra~\cite{Foster17b}, which is a form of cancellative monoid. Example models include $(\nat, +, 0)$ and
$(\seq\,A, \cat, \langle\rangle)$. We declare the following observational variables that are used in both our UTP
theories:
\begin{itemize}[]
  \item $ok, ok' : \Bool$ -- indicate divergence in the prior and present relation;
  \item $wait, wait' : \Bool$ -- indicate quiescence in the prior and present relation;
  \item $\state, \state' : \Sigma$ -- the initial and final state;
  \item $tr, tr' : \tset$ -- the trace of the prior and present relation.
\end{itemize}
Since the theory is extensible, we also allow further observational variables to be added, which are denoted by the
variables $r$ and $r'$.

\subsection{Healthiness Conditions}

We first describe the healthiness conditions of reactive relations.
\begin{definition}[Reactive Relation Healthiness Conditions]
\begin{align*}
  \healthy{R1}(P) & ~\defs~ P \land tr \le tr' \\
  \healthy{R2}_c(P) & ~\defs~ \conditional{P[\tempty, tr' \tminus tr/tr, tr']}{tr \le tr'}{P} \\
  \healthy{RR}(P) & ~\defs~ \exists (ok, ok', wait, wait') @ \healthy{R1}(\healthy{R2}_c(P))
\end{align*}
\end{definition}
$\healthy{RR}$ healthy relations do not refer to $ok$ or $wait$ and have a well-formed trace associated with them. The
latter is ensured by the reactive process healthiness conditions~\cite{Hoare&98,Cavalcanti&06,Foster17b}, $\healthy{R1}$
and $\healthy{R2}_c$, which justify the existence of the trace pseudo variable $\trace \defs tr' - tr$. $\healthy{RR}$
is closed under relational calculus operators $\false$, $\lor$, $\land$, and $\relsemi$, but not $\true$, $\neg$,
$\implies$, or $\II$. We therefore define healthy versions below.
\begin{definition}[Reactive Relation Operators]
\vspace{-1.2ex}

\begin{minipage}{.3\textwidth}
  \begin{align*}
    \truer  &~\defs~ \healthy{R1}(true) \\
    \negr P &~\defs~ \healthy{R1}(\neg P) \\
    P \rimplies Q &~\defs~ \negr P \lor Q
  \end{align*}
\end{minipage}
\begin{minipage}{.65\textwidth}
  \begin{align*}
    P \wpR Q &~\defs~ \negr (P \relsemi (\negr Q)) \\
    \IIr &~\defs~ (tr' = tr \land \state' = \state \land r' = r)
  \end{align*}
\end{minipage}
\end{definition}
We define a reactive complement $\negr P$, reactive implication $P \rimplies Q$, and reactive true $\truer$, which with
the other connectives give rise to a Boolean algebra~\cite{Foster17c}. We also define the reactive skip $\IIr$, which is
the unit of $\relsemi$, and the reactive weakest precondition operator $\wpR$. The latter is similar to the standard UTP
definition of weakest precondition~\cite{Hoare&98}, but uses the reactive complement.

We next define the healthiness conditions of reactive contracts.
\begin{definition}[Reactive Designs Healthiness Conditions]

\begin{minipage}{.2\textwidth}
\begin{align*}
  \healthy{R3}_h(P) & ~~\defs~~ \conditional{\IIsrd}{wait}{P} \\
  \healthy{RD1}(P) &~~\defs~~ ok \rimplies P \\
  \healthy{RD2}(P) &~~\defs~~ P \relsemi \ckey{J} \\
  \healthy{RD3}(P) &~~\defs~~ P \relsemi \IIsrd
\end{align*}
\end{minipage}
\begin{minipage}{.6\textwidth}
\begin{align*}
  \IIsrd            & ~~\defs~~ \healthy{RD1}(\conditional{(\exists \state @ \IIr)}{wait}{\IIr}) \\
  \healthy{R}_s     & ~~\defs~~  \healthy{R1} \circ \healthy{R2}_c \circ \healthy{R3}_h \\
  \healthy{SRD}(P) &~~\defs~~ \healthy{RD1} \circ \healthy{RD2} \circ \healthy{R}_s \\
  \healthy{NSRD}(P) &~~\defs~~ \healthy{RD1} \circ \healthy{RD3} \circ \healthy{R}_s
\end{align*}
\end{minipage}
\end{definition}
$\healthy{R3}_h$ states that if the predecessor is waiting then a reactive design behaves like $\IIsrd$, the reactive
design identity. $\healthy{RD1}$ is analagous to $\healthy{H1}$ from the theory of
designs~\cite{Hoare&98,Cavalcanti&06}, and introduces divergent behaviour: if the predecessor is divergent ($\neg ok$),
then a reactive design behaves like $\truer$ meaning that the only observation is that the trace is
extended. $\healthy{RD2}$ is identical to $\healthy{H2}$ from the theory of
designs~\cite{Hoare&98,Cavalcanti&06}. $\healthy{RD3}$ states that $\IIsrd$ is a right unit of sequential
composition. $\healthy{R}_s$ composes the reactive healthiness conditions and $\healthy{R3}_h$. We then finally have the
healthiness conditions for reactive designs: $\healthy{SRD}$ for ``stateful reactive designs'', and $\healthy{NSRD}$ for
``normal stateful reactive designs''.

Next we define the reactive contract operator.
\begin{definition}[Reactive Contracts]
  \begin{align*}
    \design{P}{Q} &~~\defs~~ (ok \land P) \implies (ok' \land Q) \\
    P \wcond Q    &~~\defs~~ \conditional{P}{wait'}{Q} \\
    \rc{\!P\!}{\!Q\!}{\!R}  &~~\defs~~ \ckey{R}_s(P \vdash Q \diamond R)
  \end{align*}
\end{definition}
A reactive contract is a ``reactive design''~\cite{Cavalcanti&06,Oliveira&09}. We construct a UTP design~\cite{Hoare&98}
using the design turnstile operator, $\design{P}{Q}$, and then apply $\healthy{R}_s$ to the resulting construction. The
postcondition of the underlying design is split into two cases for $wait'$ and $\neg wait'$, which indicate whether the
observation is quiescent, and correspond to the peri- or postcondition.

Finally, we define the healthiness conditions that specialise our theory to stateful-failure reactive designs.
\begin{definition}[Stateful-Failure Healthiness Conditions]
  \begin{align*}
    \Skip             &~~\defs~~ \rc{\truer}{\false}{\trace = \snil \land \state' = \state} \\
    \healthy{CSP3}(P) &~~\defs~~ \Skip \relsemi P \\
    \healthy{CSP4}(P) &~~\defs~~ P \relsemi \Skip \\
    \healthy{NCSP}(P) &~~\defs~~ \healthy{NSRD} \circ \healthy{CSP3} \circ \healthy{CSP4}
  \end{align*}
\end{definition}
$\Skip$ is similar to $\IIsrd$, but does not refer to $ref$ in the postcondition. If $P$ is $\healthy{CSP3}$ healthy
then it cannot refer to $ref$. If $P$ is $\healthy{CSP4}$ healthy then the postcondition cannot refer to $ref'$, but the
pericondition can: refusals are only observable when $P$ is quiescent~\cite{Hoare85,Cavalcanti&06}.


\end{document}